\documentclass[10pt]{article}
\usepackage[pdftex]{hyperref}
\newif\ifFOCS \FOCSfalse

\usepackage{verbatim}
\usepackage[cmex10]{amsmath}
\usepackage{amssymb}
\usepackage{amsthm}
\renewcommand{\Pr}{\mathbf{Pr}}

\newcommand{\E}{\mathbf{E}}

\newcommand{\N}{\mathcal{N}}
\newcommand{\D}{\mathcal{D}}
\newcommand{\G}{\mathcal{G}}

\newcommand{\Matching}{\texttt{Matching}}
\newcommand{\Sample}{\texttt{Sample}}
\newcommand{\FlowAndCut}{\texttt{FlowAndCut}}

\newcommand{\R}{\mathbb{R}}

\renewcommand{\u}{\mathbf{u}}
\newcommand{\U}{\mathbf{U}}

\renewcommand{\L}{\mathcal{L}}
\newcommand{\one}{\mathbf{1}}
\newcommand{\w}{\mathbf{w}}

\newcommand{\V}{V}
\newcommand{\Stretch}{\mathsf{Stretch}}
\newcommand{\Ball}{\mathsf{Ball}}

\renewcommand{\v}{\mathbf{v}}

\newcommand{\eps}{\varepsilon}
\ifFOCS
\newcommand{\FOCSspace}[1]{}
\newcommand{\mybox}[1]{\medskip \noindent\fbox{\parbox{0.97\columnwidth}{#1}}\medskip}
\else
\newcommand{\FOCSspace}[1]{}
\newcommand{\mybox}[1]{\medskip \noindent\fbox{\parbox{0.98\textwidth}{#1}}\medskip}
\fi
\newcommand{\grant}{Research supported by a UC Berkeley Regents Fellowship and in part by
NSF grant CCF-0635401}
\newtheorem{theorem}{Theorem}[section]
\newtheorem{claim}[theorem]{Claim}
\newtheorem{defn}[theorem]{Definition}
\newtheorem{lemma}[theorem]{Lemma}

\begin{document}
\ifFOCS
\else
\title{Breaking the Multicommodity Flow Barrier for $O(\sqrt{\log n})$-Approximations to
Sparsest Cut}
\author{Jonah Sherman\thanks{\grant}\\University of California, Berkeley}
\fi
\maketitle
\thispagestyle{empty}
\begin{abstract}
This paper ties the line of work on algorithms that find an $O(\sqrt{\log
n})$-approximation to the {\sc sparsest cut} together with the line of work
on algorithms that run in sub-quadratic time by using only single-commodity flows.
We present an algorithm that simultaneously achieves both goals, finding an
$O(\sqrt{\log(n)/\eps})$-approximation using
$O(n^{\eps}\log^{O(1)}n)$ max-flows.
The core of the algorithm is a stronger, algorithmic version of Arora \emph{et al.}'s
structure theorem, where we show that matching-chaining argument at the heart of their
proof can be viewed as an algorithm that finds good augmenting paths in certain geometric
multicommodity flow networks.  By using that specialized algorithm in place of
a black-box solver, we are able to solve those instances much more efficiently.

We also show the cut-matching game framework can not achieve an approximation
any better than $\Omega(\log(n)/\log\log(n))$ without re-routing flow.
\end{abstract}
\section{Introduction}
We consider the problem of partitioning a graph into
relatively independent pieces in the sense that not too many edges cross between them.
Two concrete optimization problems arising in that context
are the {\sc sparsest cut} and {\sc balanced separator} problems.  We are given an undirected
weighted graph $G$ on $n$ vertices, where each edge $xy$ has capacity $G_{xy}$
(we identify a graph with its adjacency matrix).
The \emph{edge expansion} of a cut $(S,\overline{S})$ is $h(S) =
\frac{\sum_{x \in S, y\in \overline{S}} G_{xy}}{\min\{|S|,|\overline{S}|\}}$.
The {\sc sparsest cut} problem is to find a cut $(S,\overline{S})$ minimizing $h(S)$; we write
$h(G)$ to denote the value of such a cut.  The {\sc balanced separator} problem has the
same objective but the additional constraint that $\min\{|S|,|\overline{S}|\} \geq \Omega(n)$.
Both problems are NP-hard, so we settle for approximation algorithms.

Most of the original work on graph partitioning focused on achieving the best
approximation factor and falls into one of two themes.  The first is based on
multicommodity flow, using the fact that
if a graph $H$ of known expansion can
be routed in $G$ via a feasible flow, then $h(H) \leq h(G)$.
If $H$ is some fixed graph, finding the best possible lower bound is equivalent to solving
the maximum concurrent flow problem; i.e., maximizing $\alpha$ such that $F \leq G$ and $D \geq \alpha H$
 where $D_{xy} = \sum_{p : x\leftrightarrow y}f_p$ is the \emph{demand graph} and $F_{xy}
= \sum_{p \owns xy}f_p$ is the \emph{flow graph} of the underlying flow.
By taking $H$ to be the complete graph, Leighton and Rao showed an
upper bound of $h(G) \leq O(\log n)\alpha^* h(H)$ for the optimal $\alpha^*$, yielding an $O(\log n)$ approximation.
The other theme is the discrete Cheeger's
inequality of Alon and Milman\cite{Alon} characterizing the relationship between cuts and the
spectrum of a graph's Laplacian matrix.  In particular, if $G$ has maximum degree $d$,
then $\lambda_2(\L_G)/2 \leq h(G) \leq \sqrt{2d\lambda_2(\L_G)}$, where $\lambda_2(\L_G)$
is second smallest eigenvalue of $G$'s Laplacian.
The two themes are incomparable, as the latter is a better approximation when
$G$ is an expander (i.e., $h(G)/d$ is large) while the former is better when $G$ has
sparse cuts.

Arora, Rao, and Vazirani naturally combined the two themes.
Rather than embedding a fixed graph $H$ of known expansion, they embed an arbitrary $H$
and then certify $H$'s expansion via $\lambda_2(\L_H)$\cite{ARV}.
Since $\lambda_2(\L_H) \geq \alpha$ is equivalent to $\L_H
\succeq \frac{\alpha}{n} \L_K$, where $K$ is the complete graph, the problem of finding
the best such lower-bound can be cast as a semidefinite program:
\begin{equation}
\max \alpha \qquad s.t. \quad \frac{\alpha}{n}\L_K \preceq \L_D, \quad F \leq G\label{SDP}
\end{equation}
They showed that for the optimal $\alpha^*$, one has an upper bound of $h(G)
\leq O(\sqrt{\log n}) \alpha^*$, yielding the currently best known approximation
factor.
Shortly thereafter, Arora, Hazan, and Kale designed a primal-dual algorithm to
approximately solve \eqref{SDP} in $\tilde{O}(n^2)$ time using multicommodity
flows\cite{AHK}.

More recently, researchers have focused on designing efficient algorithms for graph partitioning
that beat the quadratic multicommodity flow barrier. Khandekar, Rao, and Vazirani designed
a simple primal-dual framework for constructing such algorithms based on the
\emph{cut-matching game} and showed one could achieve an $O(\log^2n)$ approximation in
that framework using polylog max-flows\cite{KRV}.
Arora and Kale designed a very general primal-dual
framework for approximately solving SDPs\cite{AroraKale}.  They showed efficient algorithms
for several problems could be designed in their framework, including an $O(\log n)$-approximation to
{\sc sparsest cut} using polylog max-flows.  They also showed one could achieve
an $O(\sqrt{\log n})$-approximation in their framework using multicommodity flows,
simplifying the previous algorithm of \cite{AHK}.
Orecchia \emph{et al.} extended the cut-matching game framework of \cite{KRV} to achieve
an $O(\log n)$ approximation\cite{OSVV}.  They present two slightly
different algorithms, and remarkably, their second algorithm is the same as Arora
and Kale's, even though they never explicitly mention any SDP.
They also showed a lower bound of
$\Omega(\sqrt{\log n})$ on the approximation factor achievable in the cut-matching
framework, suggesting the framework might precisely capture the limits of current
approximation algorithms and posed the question of whether $O(\sqrt{\log n})$ could be
efficiently achieved in that framework.

\subsection{This Paper.} We tie those two lines of work together by simultaneously
achieving the $O(\sqrt{\log n})$ approximation factors of the former with the nearly
max-flow running time of the latter.
\begin{theorem} \label{alg} For any $\eps \in [O(1/\log(n)), \Omega(1)]$, there is an algorithm to
approximate the {\sc sparsest cut} and
{\sc balanced separator} problems to within a factor of
$O(\sqrt{\log(n)/\eps})$ using only $O\left(n^\eps\log^{O(1)}(n)\right)$ max-flows.
\end{theorem}
Theorem \ref{alg} effectively subsumes the results of \cite{AHK, KRV, AroraKale,
OSVV}, as taking $\eps = \Theta(1/\log(n))$ yields an $O(\log(n))$ approximation using
polylog max-flows, while any constant $\eps < 1/2$ achieves an $O(\sqrt{\log(n)})$
approximation in sub-quadratic $\tilde{O}(m + n^{3/2+\eps})$ time using the max-flow
algorithm of Goldberg and Rao\cite{GoldbergRao}.
We also show the cut-matching game framework of \cite{KRV} can not achieve an approximation
better than $\Omega(\log(n)/\log\log(n))$ without re-routing flow.

We build heavily on Arora and Kale's work,
achieving our improvement by replacing their use of
a black-box multicommodity flow solver with a specialized one
that makes use of the additional structure present in the flow instances that arise.
We begin in section \ref{FLOWS} by reviewing the nature of those flow problems,
as well as the main ideas behind the algorithms of \cite{AHK, AroraKale, OSVV}.
Having clarified the connection to partitioning, we also state our main technical result,
theorem \ref{ARVstrong}.
In section \ref{ALG} we describe the details of our algorithm, the correctness of which
follows immediately from theorem \ref{ARVstrong}.
The proof of theorem \ref{ARVstrong} appears in section
\ref{PROOF}.  Our lower-bound for the cut-matching game is then discussed in
\ref{CUTMATCH}, and we finish with some concluding remarks in section \ref{CONCLUSION}.

\section{\label{FLOWS}Expander Flows}
Expander-flow based algorithms all work by approximately solving \eqref{SDP}, either explicitly
as in \cite{AHK, AroraKale}, or implicitly as in \cite{KRV, OSVV}, by iteratively
simulating play of its corresponding two-player zero-sum game.
The game has two players: the embedding player and the flow player.
The embedding player chooses a non-trivial embedding $\V = (\v_1,\ldots,\v_n) \in
(\R^{d})^n$ of the vertices of $G$.
The flow player chooses a feasible flow $F \leq G$ supporting demands $D$ with the goal
of routing flow between points that are far away in the embedding.
More precisely, the payoff to the flow player is:
\begin{equation} \Phi(\V, D) =
\frac{\sum_{x<y}D_{xy}\|\v_x-\v_y\|^2}{\frac{1}{n}\sum_{x<y} \|\v_x - \v_y\|^2 } \notag
\end{equation}
For given demands $D$, the best
response for the embedding player is the one-dimensional embedding
given by an eigenvector of $\L_D$ of eigenvalue $\lambda_2(\L_D)$, yielding a value of
$\lambda_2(\L_D)$.
On the other hand, for a given embedding,
the best response for the flow player is a solution to the
weighted maximum multicommodity flow problem given by
\begin{equation}
\max \sum_{x<y}D_{xy}\|\v_x-\v_y\|^2 \qquad s.t. \quad F \leq G \label{MCF}
\end{equation}

The frameworks of \cite{KRV, AroraKale, OSVV} start with an initial embedding $\V^1$, such
as all points roughly equidistant.  On a given iteration $t$, the algorithm presents $\V^t$ to
the flow player, who must either respond with demands $D^t$ of value $\Phi(\V^t,D^t) \geq
1$, or a cut $C^t$ of expansion at most $\kappa$, where $\kappa$ is the desired approximation
factor.  In the latter case the algorithm terminates; in the former case, the demands
are used to update the embedding for the next iteration.
The precise update differs among each algorithm, but essentially vertices $x,y$ with large
$D^t_{xy}$ will be squeezed together in the embedding.  The analysis of
\cite{KRV, AroraKale, OSVV} show that after $T$ iterations, for
sufficiently large $T$, their adaptive strategies actually played nearly as well as they
could have in hindsight, in that
\begin{equation}
\lambda_2\left(\frac{\L_{D^1} + \cdots + \L_{D^T}}{T}\right) \geq \Omega(1) \notag
\end{equation}
Since averaging $T$ feasible flows yields a feasible flow, after $T$ iterations
the graph $D = (D^1 + \cdots + D^T)/T$ with $\lambda_2(\L_D) \geq \Omega(1)$ has
been routed in $G$.  Thus, for a given graph, the algorithm either routes an
$\Omega(1)$-expander-flow in $G$ or else finds a cut of expansion $\kappa$.
Using a binary search and scaling the edge capacities appropriately yields
an $O(\kappa)$ approximation algorithm.

The embedding can be updated in nearly linear time, and $T = O(\log^{O(1)}(n))$,
so the running time of such algorithms is dominated by the running time of the flow player.
By sparsifying $G$ (using e.g. \cite{Sparse}), we can and shall assume it has
$m = O(n\log n)$ edges.
Using Fleischer's multicommodity flow algorithm\cite{Fleischer} as a black box,
a nearly optimal pair of primal/dual solutions to \eqref{MCF} can be computed in
$\tilde{O}(n^2)$ time.  Note that \eqref{MCF} has demand weights for
every pair of vertices, so $\Omega(n^2)$ space is required to even explicitly
write it down.  On the other hand, each $\v_x \in \R^{O(\log n)}$, so
the weights $\|\v_x-\v_y\|^2$ are all implicitly stored in only $O(n\log n)$ space.
Therefore, making use of the additional geometric structure of these instances is crucial to
achieving sub-quadratic time.  Implicit in all of \cite{KRV, AroraKale, OSVV} is a
specialized algorithm to approximately solve \eqref{MCF}.  The actual algorithm
used is the same in all three, and those algorithms differ only in their strategy for
the embedding player.

In the next two subsections, we briefly sketch the single-commodity and
multicommodity flow based algorithms of \cite{AroraKale}, and then describe how we tie the two together.
In particular, our algorithm is essentially an ``algorithmetization'' of the
multicommodity flow algorithm's analysis.  For the rest of the section, suppose we have an embedding $\V$ with
$\sum_{x < y} \|\v_x-\v_y\|^2 = n^2$, and let us further assume that the points are unique
and $\|\v_x\| \leq 1$ for all $x$; i.e., the diameter is not much more
than the average distance.
\subsection{Using Single-Commodity Flows}
Consider first the absolute simplest case, where $d=1$ and the points are
simply numbers in $[-1,1]$.
It is easy to see that since the points are in $[-1,1]$, unique, and have average
squared-distance $\Omega(1)$, there must be some interval $[a,b]$, where $b-a =
\Omega(1) =: \sigma$ and the set of points to the left of $a$, $A = \{x : \v_x \leq a\}$ and to the
right of $b$, $B = \{y : \v_y \geq b\}$ have $|A| = |B| = \Omega(n) =: 2cn$.
A natural way to try to push flow far along this line would be to shrink $A$ and $B$
down to single vertices and then compute a max-flow from $A$ to $B$.

\mybox{
\FlowAndCut$(\kappa, c, w_1,\ldots,w_n \in \R)$:
\begin{itemize}
\item Sort $\{w_x\}$, let $A$ be the $2cn$ nodes $x$ with least $w_x$ and $B$ be those with
greatest $w_y$.
\item Add two vertices $s,t$.  Connect $s$ to each $x \in A$ and $t$ to each $y \in B$ with edges of capacity
$\kappa$.
\item Output the max-flow/min-cut for $s-t$.
\end{itemize}
}

Consider invoking \FlowAndCut$(\kappa, c, \v_1,\ldots,\v_n)$ with $\kappa =
c^{-1}\sigma^{-2}$.
If the max-flow is at least $\kappa cn$, then since all flow must cross the gap $[a,b]$,
we have pushed $\kappa c n$ units of flow across a squared-distance of $\sigma^2$,
achieving a solution $D$ with $\Phi(\V,D) \geq (\kappa c n \sigma^2)/n = 1$.
Otherwise, if the min-cut is at most $\kappa c n$, then at most $cn$ of the added
$\kappa$-capacity edges are cut, so at least $cn$ vertices must remain on each side and
the cut has expansion at most $\kappa$.
That is, for dimension one a $\kappa = O(1)$ approximation is obtained.

The approach of \cite{AroraKale, OSVV}
is to reduce the general case to the one-dimensional case by
picking a random standard normal vector $\u$ and projecting each $\v_x$ along
$\u$, yielding the 1-dimensional embedding $w_x = \v_x \cdot \u$.
The fact that the points are in the unit ball and have average distance
$\Omega(1)$ implies that with probability $\Omega(1)$, there is a gap $[a,b]$
with $b-a = \Omega(1) = \sigma$ as before.  Applying the previous analysis, we either
find a cut of expansion $O(1)$ or a flow with
$\sum_{x<y}D_{xy}(w_x-w_y)^2 \geq n$.  Then, the Gaussian tail
ensures that distances could not have been stretched too much along $\u$:
with high probability $(w_x - w_y)^2 \leq O(\log n)\|\v_x - \v_y\|^2$
for every pair $x,y$.  Thus, $\Phi(\V,D) \geq \Omega(1/\log(n))$, yielding an
$O(\log n)$ approximation.

\subsection{Using Multi-Commodity Flows}
Arora, Rao, and Vazirani showed that, if a \emph{best} response $D^*$ to $\V$ has
$\Phi(\V,D^*) \leq 1$, one can find a cut of expansion $O(\sqrt{\log n})$.
Supposing the optimal solution to \eqref{MCF} has value at most $n$, there must be
a solution to the dual problem of value at most $n$.
The dual assigns lengths $\{w_e\}$ to the edges of $G$, aiming to minimize $\sum_{e}G_e w_e$
subject to the constraints that the shortest-path distances between each $x,y$ under $\{w_e\}$
are at least $\|\v_x-\v_y\|^2$.
Arora and Kale show the existence of such a dual solution implies that
projecting the points along a random $\u$ and running
\FlowAndCut with $\kappa = \Theta(\sqrt{\log n})$ must yield a cut
of capacity at most $\kappa cn$ with probability $\Omega(1)$.

If not, then a flow of value at least $\kappa cn$ is returned for $\Omega(1)$ of the
directions $\u$ along which $A$ and $B$ are $\sigma$-separated.
For simplicity, assume that the flows actually correspond to a matching between $A$
and $B$.  That is, each $x$ has either zero flow leaving, or else has exactly $\kappa$
flow going to a unique $y$ along a single path.  On the one hand, that matching is routed in $G$ along
$cn$ flowpaths, each carrying flow $\kappa$.  On the other hand,
the total volume of $G$ is only $\sum_{e} G_e w_e = n$, so $\Omega(n)$ of those
flowpaths must have length at most $O(1/\kappa)$ under $\{w_e\}$.

For each $\u$, let $M(\u)$ be the matching consisting of those demand pairs routed along
such short paths.  Then, according to the following definition, $M$ is an
$(\Omega(1),\Omega(1))$-matching-cover.
\begin{defn} A $(\sigma, \delta)$-\textbf{matching-cover} for an embedding $\{\v_x\}$
is a collection $\{M(\u)\}_{\u \in \R^d}$ of directed matchings satisfying the following
conditions.
\begin{itemize}
\item Stretch: $(\v_y - \v_x) \cdot \u \geq \sigma$ for all $(x,y) \in M(\u)$
\item Skew-symmetry: $(x,y) \in M(\u)$ iff $(y,x) \in M(-\u)$
\item Largeness: $\E_\u\left[|M(\u)|\right] \geq \delta n$
\end{itemize}
For a list of vectors $\u_1,\ldots,\u_R$, let
$M(\u_1,\ldots,\u_R)$ denote the graph that contains edge $(x,y)$ iff there exist
$x_0,\ldots,x_R$ with $x_0=x,x_R = y$ and $(x_{r-1},x_r) \in M(\u_r)$ for all $r \leq R$.
For the empty list, let $M()$ denote the graph where each vertex has a directed self-loop.
Note that $M(\u_1,\ldots,\u_R)$ is not a matching, but rather a graph with maximum
in-degree and out-degree one.
\end{defn}
Furthermore, $M$ has the property that for each edge $(x,y) \in M(\u)$, the distance between
$x$ and $y$ under $\{w_e\}$ is at most $O(1/\kappa)$.  The following theorem holds for
$M$.
\begin{theorem}[\cite{JamesLee}, refining \cite{ARV}]\label{ARV}
Let $M$ be a $(\Omega(1), \Omega(1))$-matching-cover for $\{\v_x\}$.  Then, there
are vertices $x,y$ and $\u_1,\ldots,\u_R$ where $R \leq O(\sqrt{\log n})$ such that $(x,y) \in
M(\u_1,\ldots,\u_R)$ and $\|\v_x - \v_y\|^2 \geq L$.
\end{theorem}
In other words, there are vertices $x,y$ with $\|\v_x - \v_y\|^2 \geq L$ that are only $R$
matching hops away in $M$.  Applying theorem \ref{ARV}, there are vertices $x,y$ with
$\|\v_x-\v_y\|^2 \geq L$ but of distance only $RO(1/\kappa)$ under
$\{w_e\}$.  Choosing $\kappa = O(R/L) = O(\sqrt{\log n})$ yields a contradiction to the
assumption that $\{w_e\}$ is dual feasible.

\subsection{Results.}
Our improvement comes from being able to achieve an $O(\sqrt{\log n})$ gap between cut and
flow solutions, as in the latter case, while still only using single-commodity flows, as
in the former case.  
\ifFOCS
\else Recall the case of $d > 1$
was reduced to the $d = 1$ case by projecting along a random vector and
bounding the squared-stretch by $O(\log n)$.  Indeed,
the stretch could be nearly that much, so
simply pushing flow along a single direction will not allow
us to achieve anything better; in fact, that is the main idea behind our lower bound for
the cut-matching game. \fi

To do better, we need to do something more sophisticated than simply push flow
along a single direction.  A natural idea is to try to pick several directions
$\u_1,\ldots,\u_R$, push flow along each of them, and then try to glue the flows together
to actually push flow far away globally.  One motivation for such an approach is that
it seems to be the next simplest thing to do,
following that of using only a single direction.  The second and most crucial motivation is
to observe that \emph{such an approach is strongly suggested by the analysis for the
multicommodity flow algorithm just sketched}.
\ifFOCS \else
To see that, suppose the typical
flowpath along a random $\u$ routes between points of squared-distance $\Delta$.
Theorem \ref{ARV} says we can always augment $R = O(\sqrt{\log n})$ such flowpaths to
route demand between points of squared-distance $L = \Omega(1)$, at the cost of possibly raising
congestion by a factor of $R$.  Thus, either $\Delta \geq L/R = \Omega(1/\sqrt{\log n})$,
or else augmenting together $R$ typical flowpaths and scaling down by $R$ maintains
feasibility and increases the objective of \eqref{MCF}.
\fi

Unfortunately, theorem \ref{ARV} doesn't say anything at all about \emph{finding} such
directions $\u$, or whether the same $\u_1,\ldots,\u_R$ will
simultaneously work for many vertices.  To analyze such an algorithm, we need a stronger,
algorithmic version of theorem \ref{ARV}.  Our main technical contribution is such a
theorem.
\begin{theorem} \label{ARVstrong} For any $1 \leq R \leq \Theta(\sqrt{\log(n)})$, there is
$L \geq \Theta(R^2/\log(n))$ and an (efficiently sample-able) distribution $\D$ over
$(\R^d)^{\leq R}$ with the following property.

If $M$ is an $(\Omega(1), \Omega(1))$-matching-cover for
$\{\v_x\}$, then the expected number of edges $(x,y)$ with $(x,y) \in M(\D)$ and
$\|\v_x-\v_y\|^2 \geq L$ is at least $e^{-O(R^2)}n$.
\end{theorem}

Using theorem \ref{ARVstrong} and choosing $R = \Theta(\sqrt{\eps \log(n)})$, $L =
\Theta(\eps)$, we simply sample $\u_1,\ldots,\u_R$ from $\D$, and let our final flowpaths be the
concatenation of those along $\u_1,\ldots,\u_R$.
On average, we get $n^{1-\eps}$ such paths, and thus
have \emph{simultaneously} routed $n^{1-\eps}$ paths between points of squared distance $L$
using only $R$ single-commodity flows.
Using an iterative re-weighing scheme and
repeating $O(n^{\eps}\log^{O(1)}n)$ times, we achieve a feasible flow and an
approximation ratio of $O(R/L) = O(\sqrt{\log(n)/\eps})$.

That is, to push flow far away, we sample $\u_1,\ldots,\u_R$ from $\D$ and then
iteratively push flow along each direction.  The distribution $\D$ will essentially
consist of picking a random direction $\u_1$, and then choosing $\u_{r+1}$ to be a
$1-1/R$-correlated copy of $\u_{r}$; i.e., a vector extremely close to $\u_{r}$.
Because $\u_{r+1}$ and $\u_r$ are so close, it is intuitively clear and easy to argue that
\emph{if} flow gets pushed along
at each step, it must be pushed far away, as the
projections along each $\u_r$ will essentially add together.
The somewhat counterintuitive fact is that flow actually does get pushed further along in
this manner.
Even though $\u_{r}$ and $\u_{r+1}$ are extremely close together, a significant
fraction of vertices
that were in the ``sink set'' along $\u_{r}$ will be in the ``source set'' along $\u_{r+1}$.
That phenomenon is a consequence of measure concentration.

\section{\label{ALG}The Algorithm}
While we found it most convenient to discuss expander flows and the corresponding game
in the context of the {\sc sparsest cut} problem, our algorithm applies
most directly to {\sc balanced separator}, which has a similar SDP relaxation and game.
Roughly, the difference is that in the {\sc balanced separator} case the embedding
player must choose an embedding for which the maximum squared distance between points
is not much larger than the average.  When the average distance is $\Theta(1)$, this is
equivalent to the requirement that $\|\v_x\| \leq O(1)$ assumed earlier in section
\ref{FLOWS}.
The reduction from {\sc sparsest cut} to {\sc balanced separator} is well-known, and in
fact, the unbalanced case is ``easy'' in the sense that if the cut found is unbalanced, it
will be an $O(1)$ approximation to the sparsest cut\cite{ARV}.
In particular, Arora and Kale show that one can either obtain an $O(1)$
cut/flow gap with a single max-flow, or else reduce the problem to the balanced case by
finding $\Omega(n)$ points in a ball of radius $O(1)$ that are still spread-out within
that ball; for details, we refer the reader to \cite{AroraKale}.

The precise statement of the results sketched in section \ref{FLOWS} is the following main lemma of \cite{AroraKale}.
\begin{lemma}[\cite{AroraKale}]\label{aroraoracle}
Let $U \subseteq [n]$ be a set of nodes.  Suppose we are given vectors
$\V = \{\v_x\}_{x\in U}$ of length at most $O(1)$ such that $\sum_{x,y\in
U}\|\v_x-\v_y\|^2 = n^2$.
\begin{itemize}
\item There is an algorithm that uses $O(1)$ expected
max-flow computations and outputs either a demand graph $D$ on $U$ of max-degree
$O(\log(n))$ that is
routable in $G$ with $\Phi(\V, D) \geq 1$ or a balanced cut of expansion $O(\log n)$.
\item There is an algorithm that uses a single
multicommodity flow computation and $O(1)$ expected max-flow computations and outputs either a
demand graph $D$ on $U$ of max-degree $O(1)$ that is routable in $G$ with $\Phi(\V,D) \geq 1$ or
a balanced cut of expansion $O(\sqrt{\log n})$.
\end{itemize}
\end{lemma}
The importance of the degree is for the running time; if each $D^t$ has max-degree $\beta$,
then the total number of iterations needed is $O(\beta \log(n))$\cite{AroraKale}.
To prove theorem \ref{alg}, we replace lemma \ref{aroraoracle} with the following.
\begin{lemma} \label{oracle} Let $U$, $\V$ be as in lemma \ref{aroraoracle}.
For any $\eps \in [O(1/\log(n)), \Omega(1)]$, there is an algorithm that
uses $O(n^{\eps}\log^{O(1)}(n))$ expected max-flow computations and outputs either
a demand graph $D$ on $U$ of max-degree $O(1/\eps)$ routable in $G$ with $\Phi(\V,D) \geq 1$ or a
balanced cut of expansion $O(\sqrt{\log(n)/\eps})$.
\end{lemma}
For the rest of this section, we prove lemma \ref{oracle}.
We first immediately try to find a cut, using \FlowAndCut.  The parameters $c,\sigma$
are set by the following lemma.
\begin{lemma}[\cite{ARV}]\label{goodu} Let $U$, $\V$ be as in lemma \ref{aroraoracle}.  Then, there exist
$c,\sigma,\gamma = \Omega(1)$ so for a random $\u$, with probability at least $\gamma$
the sets $A,B$ in \FlowAndCut$(\cdot, c, \{\v_x \cdot \u\}_{x\in U})$ have $(\v_y - \v_x)
\cdot \u \geq \sigma$ for all $x\in A, y\in B$.
\end{lemma}
Let us call the $\u$ described by lemma \ref{goodu} \emph{good}, and set $\delta = \gamma
c/16$.   Let $\eps \in [O(1/\log(n)), \Omega(1)]$ be given so that $R = O(\sqrt{\eps \log n})$
yields an expected size bound of $n^{1-\eps}$ in
theorem \ref{ARVstrong}.  Set $L = \Omega(\eps)$ as in theorem \ref{ARVstrong},
$\kappa = 24R/cL$, and $\beta = 12/cL$.
The following easy lemma was sketched in section \ref{FLOWS}.
\begin{lemma}[\cite{KRV, AroraKale}]\label{cutval} If \FlowAndCut$(\kappa, c,\ldots)$ returns a
cut of capacity at most $\kappa cn$, then the cut is $cn$-balanced and has expansion at
most $\kappa$.
\end{lemma}
We sample $O(\log(n))$ independent $\u$, and run \FlowAndCut$(\kappa, c, \{\v_x
\cdot \u\})$.  If we ever find a cut of capacity at most $\kappa cn$, we immediately
output it and stop, yielding a balanced cut of expansion $\kappa =
O(\sqrt{\log(n)/\eps})$.
Otherwise, with very high probability, we are in the situation where
there are at least $\gamma/2$ good $\u$ for which a flow of value at least $\kappa cn$ is
returned.  In the latter scenario, we will find a flow $D$ with $\Phi(\V, D) \geq 1$.

\subsection{Finding a Flow}
We efficiently find a solution to the maximum multicommodity flow problem
\begin{equation}\label{primal}
\begin{split}
\max \sum_{x<y}D_{xy}\|\v_x-\v_y\|^2 \\
s.t. \quad F \leq G, \quad \max_x \deg_D(x) \leq \beta
\end{split}
\end{equation}
of value at least $n$.  The dual assigns lengths
$\{w_e\}$ to edges and $\{w_x\}$ to the vertices, with the constraint that
the shortest path distance from $x$ to $y$ under these lengths dominate $\|\v_x -
\v_y\|^2$.
\[
 \min \sum_{e} G_e w_e + \sum_x \beta w_x
\]
\[s.t. \quad
\forall p: x \leftrightarrow y\quad  w_x + w_y + \sum_{e \in p}w_e \geq \|\v_x-\v_y\|^2
\]

We use the \emph{multiplicative weights} framework to approximately solve
\eqref{primal}.
\begin{theorem}[\cite{Plotkin, Freund, MW}]\label{PST} Let $A\in \R^{m \times n}, b \in \R^m$ with $b > 0$,
and consider the following iterative procedure to find an approximate solution to $Ax \leq b$.

Initialize $y^{1} \in \R^m$ to the all-1s vector.  On
iteration $t$, query an oracle that returns $x^{t}$ such that $0 \leq Ax^{t} \leq \rho b$ and
$y^{t} \cdot Ax^t \leq y^t \cdot b$, and then update
\[ y^{t+1}_j \gets \left(1+\eta \frac{(Ax^t)_j}{\rho b_j}\right) y^t_j \]

If $0 < \eta < 1/2$, then after $T = \rho \eta^{-2}\log(n)$
iterations, $A\left(\frac{x^1 + \cdots + x^T}{T}\right) \leq (1+4\eta)b$.
\end{theorem}
We use theorem \ref{PST} with $\eta = 1/4$, initializing the dual
variables $\{w_e\}, \{w_x\}$ and updating them accordingly.  On iteration $t$, we
find a flow $(F^t,D^t)$ of objective value $2n$ that violates
the constraints by at most a factor of $\rho = O(n^{2\eps}\log n)$ and
\begin{equation}
\sum_{e} w_e F^t_e + \sum_x w_x \deg_{D^t}(x) \leq \sum_e w_eG_e + \sum_x w_x \beta
\label{relfeas}
\end{equation}
After $T = O(n^{2\eps}\log^{2}(n))$ rounds, scaling the average flow down by $2$ yields a
feasible flow of objective value $n$.
Noting that \eqref{relfeas} and the algorithm of theorem \ref{PST}
are invariant to scaling of the dual variables,
for convenience we will also scale them on each
iteration so that $\sum_{e} w_e G_e + \sum_x w_x \beta = 2n$.
In that case, any flow of objective value $2n$ that only routes along violated or tight paths
(those $p : x \leftrightarrow y$ for which $\sum_{e \in p}w_e + w_x + w_x \leq \|\v_x-\v_y\|^2$)
satisfies \eqref{relfeas}.
In our algorithm, we will only route flow along paths $p : x \leftrightarrow y$ for which
$\|\v_x-\v_y\|^2 \geq L$,
and $w_x,w_y,\sum_{e \in p}w_e \leq L/3$.

All flows will come from augmenting flows returned by \FlowAndCut, where
we identify single-commodity flows in $G \cup \{s,t\}$ with multicommodity flows
in $G$ in the obvious way.  If $F$ is an acyclic $s-t$ flow in $G\cup\{s,t\}$,
it is well-known that $F$ can be decomposed into at most $m$ flowpaths.
While computing such a decomposition could require $\Omega(nm)$ time,
fortunately we need only pseudo-decompose flows in the following sense.
\begin{defn} If $F$ is an acyclic $s-t$ flow in $G \cup \{s,t\}$ with a flow decomposition
$((f_i,p_i))_{i \leq m}$, then a list $P = ((f_i, s_i, t_i, \ell_i))_{i \leq m}$
where $p_i = s,s_i,\ldots,t_i,t$ and $\sum_{e \in p_i} w_e = \ell_i$ is a
\textbf{pseudo-decomposition} of $F$.  That is, a pseudo-decomposition is a list
containing the amount of flow, second vertex, second-to-last vertex, and length of each flowpath.
\end{defn}

The following two lemmas are easy applications of dynamic trees(see \cite{SleatorTarjan}).
\begin{lemma}\label{decomp}
Given a flow $F$ on $G\cup\{s,t\}$, a pseudo-decomposition can be computed in
$O(m\log n)$ time.
\end{lemma}
\begin{lemma}\label{scale}
Given a flow $F$ on $G\cup \{s,t\}$, and a desired scaling vector
$(\alpha_1,\ldots,\alpha_m)$, we can
compute the flow $F'$ with decomposition $\{(\alpha_k f_k, p_k)\}$ in
$O(m\log n)$ time.
\end{lemma}
Lemma \ref{scale} allows us to efficiently cherry-pick ``good'' flowpaths from the flows
returned by \FlowAndCut.

In their analysis, Arora and Kale round the flows returned by \FlowAndCut\ to
matchings.  We do the same, with a small change to ensure doing so does not
raise congestion by too much.

\mybox{
\Matching$(\u)$
\begin{itemize}
\item Call \FlowAndCut$(\kappa, c, \{\v_x \cdot \u\})$ and pseudo-decompose the
resulting flow into $P$.  Set $M = \emptyset$.
\item Throw away any $(f_i,s_i,t_i,\ell_i)\in P$ with $(\v_{t_i} - \v_{s_i})\cdot \u < \sigma$,
$f_i < \kappa cn/4m$, $w_{s_i} > L/3$, $w_{t_i} > L/3$, or $\ell_i > L/3R$.
\item Greedily match the remaining pairs: iteratively pick
$(f_i,s_i,t_i,\ell_i)\in P$, add $(s_i,t_i)$ to
$M$, and remove any $(f_j,s_j,t_j,\ell_j) \in $ with $\{s_i,t_i\} \cap \{s_j,t_j\} \neq
\emptyset$.
\item Output $M$.
\end{itemize}
}

The following lemma is essentially the same as one used in \cite{AroraKale}, and
follows by the choice of parameters.  The congestion bound, which was not needed for their
analysis but is needed for our algorithm, comes from the fact that \Matching\ discards any flowpath
with $f_i \leq \kappa cn/4m$ before scaling any remaining flows to $1$.
\begin{lemma} \label{shortmatch} \Matching\ is a $(\sigma,\delta)$-matching-cover.
Furthermore, for each $\u$, the (unit-weighted) demands \Matching$(\u)$
are simultaneously routable in $G$ with congestion at most $4m/\kappa cn$ along flowpaths of
length at most $L/3R$ under $\{w_e\}$.
\end{lemma}
\ifFOCS
\else
\begin{proof}
The symmetry and stretched properties hold by construction, so we need only establish the
largeness property.
Let $\u$ be a good direction for which the returned flow has value at least $\kappa cn$,
and let $D$ be the corresponding demands.  Since $\u$ is good, every demand pair is
$\sigma$-separated along $\u$.   Each $x \in U$ has degree at
most $\kappa$ and the total degree is at least $2\kappa cn$.
Deleting each path with $f_i \leq \kappa cn/4m$ removes at most $\kappa cn/4$ total flow.
Since $\sum_x w_x\beta \leq 2n$ and $\beta = 12/cL$, at
most $cn/4$ vertices can have $w_x > L/3$; deleting them removes at most $\kappa cn/4$
units of flow.  Finally, since the original flow was feasible in the original
graph,
\begin{equation}
 \sum_{p} f_p \ell_p = \sum_e w_e \sum_{p \owns e} f_p \leq \sum_e w_e G_e \leq 2n \notag
\end{equation}
Since $\kappa = 24R/cL$, at most $\kappa cn/4$ units can flow along paths longer than
$L/3R$.

In total, the second step of \Matching\ removes at most $3\kappa cn/4$ units of flow, so at least $\kappa
cn/2$ total degree survives.  Each greedy matching step decreases the total degree
by at most $4\kappa$, so at least $cn/8$ pairs must get matched.
Thus, the expected size of \Matching$(\u)$ is at least $(\gamma/2)(cn/8) = \delta$.

For the congestion bound, we threw away all paths with flow less than $\kappa cn/4m$, so
scaling the remaining paths to $1$ yields a flow with congestion at most $4m/\kappa
cn$.
\end{proof}
\fi
On each iteration, we sample $\u_1,\ldots,\u_R$ from the distribution $\D$ of theorem
\ref{ARVstrong} and call \Matching $(\u_r)$.  Let $D'$ be the unit-weighted
graph with an edge $(x,y)$ for each $(x,y) \in \Matching(\u_1,\ldots,\u_R)$
with $\|\v_x-\v_y\|^2 \geq L$.  By theorem \ref{ARVstrong}, the expected size of $D'$ is
at least $n^{1-\eps}$, so after $n^{\eps}/2$ expected trials, we have $|D'| \geq
n^{1-\eps}/2$.  Applying lemma \ref{scale} again $R$ times, we can compute
a flow $F'$ that routes $D'$ in $G$ with congestion $R(4m/\kappa c n) = O(\log(n)/L)$,
since $m = O(n\log n)$ by assumption.  Note also that $D'$ has max-degree $2$.

Then, $D'$ achieves an objective value of at least $|D'|L$, so
scaling up by $2n/|D'|L$ yields a solution of value $2n$
that satisfies \eqref{relfeas} and congests edges by at most an $O(n^{\eps} \log n)$ factor.
Since $\beta = 12/cL = \Omega(1/\eps)$, the degree constraints are also violated by at most an
$O(n^{\eps})$ factor.  The running time is dominated by flow computations, of which
there are an expected $O(Rn^{\eps})$ in each of $O(n^{\eps} \log^{2}(n))$ iterations,
for a total of $O(n^{2\eps}\log^{5/2}(n))$ expected max-flows.

\section{\label{PROOF}Proof of Theorem \ref{ARVstrong}}
Let $M$ be a $(\sigma, \delta)$-matching cover.  We identify $M$ with a weighted directed graph,
where edge $(x,y)$ is has weight $\Pr_\u[(x,y) \in M(\u)]$.  The skew-symmetry condition ensures
the weights of $(x,y)$ and $(y,x)$ are the same, as are the in-degree and out-degree of
each $x$.
The total out-degree of $M$ is at least $\delta n$ by assumption.
Following \cite{ARV}, we first prune $M$ to a more uniform version by iteratively removing
any vertex of out-degree less than $\delta/4$.  Doing so preserves skew-symmetry, and
at least $ \delta n/2$ out-degree remains.  It follows that we are left with a matching
cover on vertices $X$, with $|X| \geq \delta n/2$ and every $x \in X$ has out-degree at
least $\delta/4$.  The pruned $M$ is a $(\sigma,\delta/4)$-uniform-matching-cover.
\begin{defn} A
$(\sigma,\delta)$-\textbf{uniform-matching-cover} of $X \subseteq [n]$ is a $(\sigma,
0)$-matching-cover where every $x \in X$ has in-degree at least $\delta$ in $M$.
\end{defn}
\subsection{Chaining and Measure Concentration}
Let $y \in X$, and let $A$ be the set of $\u$ for which $y$ has an out-edge in $M(\u)$.
The main idea behind the proof of theorem \ref{ARV} is the following.
Since $A$ and $-A$ are two sets of measure
$\Omega(1)$, the isoperimetric profile of Gaussian space implies there must be many $\u \in A,
\hat\u \in -A$ that are very close: $\|\u - \hat\u\| \leq O(1)$ (we remark that
\cite{ARV} uses the uniform measure on the sphere, but the same analysis holds for Gaussians after scaling various
quantities by $\sqrt{d}$).
Choose $x,z$ with $(x,y) \in M(\hat\u)$, $(y,z) \in M(\u)$ and observe that
\begin{align}
(\v_y - \v_x) \cdot \u &= (\v_y-\v_x) \cdot \hat\u - (\v_y-\v_x)\cdot(\hat\u - \u) \notag\\
 &\geq \sigma - \|\v_x - \v_y\|\|\hat\u-\u\| \notag
\end{align} 
Thus, either $\|\v_x - \v_y\| \geq \Omega(\sigma)$, or else $(\v_y - \v_x) \cdot \u \geq
\sigma/2$.  In the former case, a matching edge joins two points of distance $\Omega(1)$.
In the latter case, replacing the edge $(y,z) \in M(\u)$ with $(x,z)$ yields an edge
with $(\v_z - \v_x) \cdot \u \geq (3/2)\sigma$.  By an inductive argument, the chaining case
can be repeated until an edge connects two points of distance $\Omega(1)$.
On the one hand, after $R$ chaining steps, we have pairs of points that
are $R$ matching-hops apart, $O(1)$ distance apart, and have projection $\Theta(R)$.
On the other hand, with high probability, no pair of distance $\Theta(1)$ has
projection $\Theta(\sqrt{\log n})$, so the process must end after $\Theta(\sqrt{\log n})$ steps.

To turn the argument into an algorithm, we choose a sequence of highly correlated
directions $\u_1,\ldots,\u_R$.  For $R \geq 1$ and $0\leq \rho \leq 1$, let $\N^R_\rho$ be the distribution of
$\u_1,\ldots,\u_R$ defined by choosing a standard normal $\u_1$, and then choosing
each $\u_{r+1} \sim_{\rho} \u_r$ to be a $\rho$-correlated copy of $\u_{r}$.  That is,
each of the $d$ coordinate vectors $(\u_{1,i},\ldots,\u_{R,i})$ are independently
distributed as multivariate normals with covariance matrix $\Sigma_{r,r'} =
\rho^{|r-r'|}$.
In fact, simply setting $\D = \N^R_{1-1/R}$ achieves theorem \ref{ARVstrong} for $R \leq
O(\log^{1/3}(n))$ and size bound of $e^{-O(R^3)}n$.  The barrier is essentially
the same as the one that limited the original analysis of \cite{ARV} to $R =
O(\log(n)^{1/3})$.
To overcome that barrier, we algorithmetize Lee's improvement\cite{JamesLee} by
independently sampling uncorrelated $\w_1,\ldots,\w_R$, and then shuffling the two lists
together.
The idea is that the highly correlated $\u_r$ will give us long stretch, while
the $\w_r$ will greatly increase the probability of forming a long chain, at the cost of
losing some stretch.  The sampling algorithm is:

\mybox{
\Sample$(R, \rho)$
\begin{itemize}
\item Pick $\u_1,\ldots,\u_R \sim \N^R_\rho$, $\w_1,\ldots,\w_R\sim\N^R_0$.
\item Pick a random shuffling of the two lists, pick a random $r \leq R$, and output the
first $r$ elements of the shuffled list.
\end{itemize}
}

The reason for the randomness is to keep the algorithm trivial, leaving the work to our
analysis.  We show that there exists a particular shuffling and $r \leq R$ for which
\Sample\ is good; by randomly guessing, we lose at most a $2^{R+1}$ factor in our final
expectation bound, which is negligible relative to the $e^{-O(R^2)}n$ bound we are aiming
for.

Our proof of theorem \ref{ARVstrong} closely follows Lee's proof of theorem
\ref{ARV}, the main difference being the use of a stronger isoperimetric inequality.
The the standard isoperimetric inequality says that if $A$ is a set of large
measure, then for almost points $\u$, a small ball around $\u$ has non-empty intersection
with $A$.  We use a stronger version, saying that if $A$ is a set of large
measure, then for almost all points $\u$, a small ball around $\u$ has a significantly
large intersection with $A$.
\begin{lemma}\label{iso} Let $A \subseteq \R^d$ have Gaussian measure $\delta > 0$.  If
$\u,\hat\u$ are $\rho$-correlated with $0 \leq \rho < 1$, then
\begin{align}
\Pr_{\u}\left[\Pr_{\hat\u} \left[\hat\u \in A\right] < (\eps \delta)^{1/(1-\rho)}\right]
 < \eps
\notag 
\end{align}
\end{lemma}
Lemma \ref{iso} is an easy corollary of Borell's reverse hypercontractive inequality
\cite{Borell};
\ifFOCS
\else we include a short proof in appendix \ref{ISOPROOF}. \fi
Applications of Borell's result to strong isoperimetric inequalities appear in
\cite{MORSS}, and we follow the proofs of similar lemmas there.


\subsection{Definitions}
For a matching-cover $M$ and a distribution $\D$ over $\R^*$, let $M(\D)$ be the
random graph $M(\u_1,\ldots,\u_r)$ where $\u_1,\ldots,\u_r \sim \D$.  For a random graph
$\G$ and sets $S,T \subseteq [n]$, let $\mu_\G(S,T)$ be the expected number of edges from $S$
to $T$ in $\G$.  We say $S$ is $\gamma$-\emph{connected} to $T$ in $\G$ if $\mu_\G(S,T)
\geq \gamma$.  For singleton sets, we omit braces and write $\mu_\G(x,y)$ for the probability that
the edge $(x,y)$ is in $\G$.

Two sets that will be useful are,
\begin{align}
\Ball[x;\ell] &= \{ y : \|\v_x - \v_y\| \leq \ell \} \notag \\
\Stretch[x, \sigma, \u] &= \{ y : (\v_y - \v_x) \cdot \u \geq \sigma\} \notag
\end{align}
We will also work with collections of distributions $\D = \{\D(\u)\}$ over $\R^*$
parameterized by $\u$.  Such a collection is itself associated with the distribution
induced by sampling a standard normal $\u$ and then sampling from $\D(\u)$.
\begin{defn} Let $\D$ be a distribution collection.  We say a vertex $x$ is
$(\sigma, \delta, \gamma, \ell)$-\textbf{covered}
in $M(\D)$ if for least $\delta$ of $\u$, $x$ is $\gamma$-connected to $\Stretch[x, \sigma, \u]
\cap \Ball[x;\ell]$ in $M(\D(\u))$.
\end{defn}
\subsection{Cover Lemmas}
Our goal is to exhibit a distribution $\D$ such that many vertices $x$ are well-connected to
$X \setminus \Ball[x;\sqrt{L}]$ in $M(\D)$.  To do so, we inductively construct particular
distribution collections $\D^r$ such that many vertices $x$ are either $e^{-O(Rr)}$-connected to
$X\setminus \Ball[x;\sqrt{L}]$ in $M(\D^r)$, or else are $(\Omega(r), \Omega(1), e^{-O(rR)},
\sqrt{L})$-covered by $M(\D^r)$.

We begin with a trivial bound on how much points can be covered.
\begin{lemma}\label{toolong} For $\ell, \gamma, \delta > 0$ and arbitrary $M, \D$,
no vertex is $(\ell\sqrt{2\log(n/\delta)},\delta,\gamma,\ell)$-covered by $M(\D)$.
\end{lemma}
 \begin{proof} For any $y \in
\Ball[x;\ell]$, the probability that $(\v_y - \v_x)\cdot \u \geq \beta$ is at most
$\exp(-\beta^2/2\ell^2) \leq \delta/n$ for $\beta = \ell\sqrt{2\log(n/\delta)}$.
It follows that the probability that $\Stretch[x,\ell\sqrt{2\log(n/\delta)}, \u] \cap
\Ball[x;\ell]$ is
non-empty is at most $(n-1)\delta/n < \delta$.
\end{proof}

The next lemma says that if a vertex $x$ is connected by $\D'$ to a set $S$ of vertices
that are covered by $\D$, then $x$ is covered by the concatenation of $\D'$ and $\D$.
\begin{lemma}\label{nearby} Let $S$ be a set of vertices such that each $y \in S$ is
$(\sigma,\delta,\gamma,\ell)$-covered by $M(\D)$.  Let $x$ be a vertex with
$\mu_{M(\D')}(x,
S \cap \Ball[x;\ell']) \geq \gamma'$.
Then, $x$ is $(\sigma-\sqrt{2\ell'\log(2/\delta)},\delta/4, \gamma \gamma'\delta/4, \ell +
\ell')$-covered by $\D''(\u) = \D', \D(\u)$.
\end{lemma}
\begin{proof} Let $\Gamma$ be the distribution of $x$'s out-neighbor in $M(\D')$,
conditioned on $S \cap \Ball[x;\ell']$.  For each $y \in S$, let $A_y$ be the set of $\u$ for which
$y$ is $\gamma$-connected to $\Stretch[y, \sigma, \u] \cap \Ball[y;\ell]$ in $M(\D(\u))$.

For any fixed $y \in \Gamma$, the quantity $(\v_x - \v_y) \cdot \u$ is normal with mean
zero and variance
$\|\v_y - \v_x\|^2 \leq \ell'^2$, so the probability (over $\u$) that $y \in
\Stretch[x, -\beta, \u]$ is at least $1-\exp(-\beta^2/2\ell'^2) \geq 1-\delta/2$ for
$\beta = \sqrt{2\ell'\log(2/\delta)}$.  Then, for at least
$\delta/2$ of $\u$, we have $y \in \Stretch[x, -\beta, \u]$ and
$\u \in A_y$.  By averaging, for at least $\delta/4$ of $\u$, at least $\delta/4$ of $y \sim
\Gamma$ have $y \in
\Stretch[x,-\beta, \u]$ and $\u \in A_y$.  It follows that $x$ is
$(\sigma-\beta, \delta/4, \gamma \gamma'\delta/4, \ell +
\ell')$-covered by $\D''$.
\end{proof}

Our next lemma is the main chaining step.
\begin{lemma}\label{extend} Let $M$ be a $(\sigma_0,\cdot)$-matching-cover, $T$ be a set
of vertices that are
$(\sigma, 1-\delta/2, \gamma, \infty)$-covered in $M(\D)$, and $S$ a set of vertices
that is $\delta|T|$-connected to $T$ in $M$.  Then, at least $\delta|T|/2$ vertices $x \in S$ are
$(\sigma + \sigma_0, \delta|T|/4|S|, \gamma, \infty)$-covered by $\D'(\u) = \u, \D(\u)$.
\end{lemma}
\begin{proof}
  Let $M'$ be the subgraph of $M$ consisting only of edges from
$S$ to $T$; by assumption the total degree in $M'$ is at least $\delta |T|$.  Further remove
any edge $(x,y) \in M'(\u)$ where $\mu_{M(\D(\u))}(y, \Stretch[y, \sigma, \u]) < \gamma$.
The total in-degree remaining is at least $\delta|T|/2$, so there is a set $S' \subseteq
S$ of at least $\delta|T|/2$ vertices that have out-degree at least $\delta|T|/4|S|$.
Finally, note that if $(x,y) \in M'(\u)$, then $\Stretch[y,
\sigma, \u, \infty] \subseteq \Stretch[x, \sigma+\sigma_0, \u, \infty]$, so each $x \in S'$
is $(\sigma + \sigma_0, \delta|T|/4|S|, \gamma, \infty)$-covered by $\D'$.
\end{proof}

To apply lemma \ref{extend}, we need to establish covers with $\delta$ very close to $1$.
Consider taking a collection $\D$ and then \emph{smoothing} it by replacing $\D(\u)$
with the average of $\D(\hat\u)$ for nearby $\hat\u$.
The next lemma shows that doing so boosts $\delta$ to nearly $1$, in exchange
for a loss in $\sigma$ and $\gamma$.
\begin{lemma}\label{boost} Let $x$ be $(\sigma,\delta,\gamma,\ell)$-covered by $\D$.
Then, $x$ is $(\rho\sigma - 4\ell \sqrt{\log(2/\delta)}, 1-2\delta,
\delta^{2/(1-\rho)}\gamma/4, \ell)$-covered by $\D'(\u) =
\D(\hat\u)$ where $\hat\u \sim_\rho \u$.
\end{lemma}
\begin{proof}
Let $A$ be the set of $\hat\u$ for which $x$ is $\gamma$-connected to $\Stretch[x, \sigma, \hat\u] \cap
\Ball[x;\ell]$ in $M(\D(\hat\u))$.
For each $\hat\u$, let $\Gamma(\hat\u)$ be the distribution of $x$'s out-neighbor in
$M(\D(\hat\u))$, conditioned on $\Stretch[x, \sigma, \hat\u] \cap \Ball[x;\ell]$.

For any $\hat \u$ and $y \in \Gamma(\hat\u)$, the quantity $(\v_y - \v_x) \cdot
\u$ is normal with mean $\rho (\v_y - \v_x) \cdot \hat \u \geq \rho \sigma$ and variance
$(1-\rho^2)\|\v_y - \v_x\|^2 \leq 2(1-\rho)\ell^2$; it follows that
$y \in \Stretch[x, \rho \sigma - \beta, \u]$ with probability at least
$1-\exp(-\beta^2/4(1-\rho)\ell^2) \geq 1-(\delta/2)^{4/(1-\rho)}$ for $\beta = 4\ell
\sqrt{\log(2/\delta)}$ over $\u$.  By
averaging, for at least $1-\delta$ $\u$, for at least
$1-(2/\delta)(\delta/2)^{4/(1-\rho)}$
$\hat\u \sim_\rho \u$, we have $\Pr\left[\Gamma(\hat\u) \in \Stretch[x, \rho\sigma-\beta,
\u]\right]
\geq 1/2$.  Call such pairs $(\u, \hat\u)$ good.

Applying lemma \ref{iso} to $A$, for at least $1-\delta$ $\u$, we have $\Pr[\hat\u \in
A] \geq \delta^{2/(1-\rho)}$.
All together, for at least $1-2\delta$ $\u$, with probability at least
$\delta^{2/(1-\rho)} - (2/\delta)(\delta/2)^{4/(1-\rho)}$ we
have both $\hat\u \in A$ and $(\u,\hat\u)$ good.  In that case, $\mu_{M(\D(\hat\u)}(x,
\Stretch[x, \rho\sigma-\beta, \u] \cap \Ball[x;\ell]) \geq \gamma/2$.
The lemma follows by noting $(2/\delta)(\delta/2)^{4/(1-\rho)} \leq
\delta^{2/(1-\rho)}/2$.
\end{proof}

Combining the previous results, we prove the main inductive lemma.
\begin{lemma}\label{induct} Let $M$ be a $(\sigma,\delta)$-uniform-matching-cover of $X$
where $\delta \leq 1/4$.  Let $\ell \leq \sigma/2^7\sqrt{\log(1/\delta)}$ and $K \geq
1$.  Then, one of the following must occur.
\begin{enumerate}
\item\label{first} There are distribution collections $\D^0,\ldots,\D^K$ such that for every $k \leq
K$, at least $\delta^{6k}|X|$ vertices are $(k\sigma/4, \delta^8, \delta^{24Kk},
\ell)$-covered in $M(\D^k)$.
\item\label{second} There is a distribution $\D^*$ such that at least $\delta^{6K}|X|$ vertices $x$
are $\delta^{24K^2}$-connected to $X \setminus B[x;\ell]$ in $M(\D^*)$.  Furthermore, $\D^*$
is a shuffling of $\N^k_{1-1/K}$ with $\N^{k'}_0$ for some $k \leq K$ and $k' \leq 6K$.
\end{enumerate}
\end{lemma}
\begin{proof}
For $k = 0$, every $x\in X$ is $(0, 1, 1, 0)$-covered by $\D^0(\u) = ()$, the
empty list.

Assuming case \ref{first} holds for some $0 \leq k < K$, let $T_0$ be those vertices that are
$(k\sigma/4, \delta^8, \gamma, \ell)$-covered by $\D^{k}$, where $\gamma =
\delta^{24Kk}$.
We begin by finding a set $S$ that is well-connected to $T_0$.
Since at least $\delta|T_0|$ in-degree enters $T_0$ in $M$, by averaging
either at least $\delta^{-1}|T_0|$ vertices have
at least $\delta^2|T_0|/|X|$ out-degree into $T_0$ or else at least $\delta|T_0|$ vertices
have at least $\delta^3$ out-degree into $T_0$.  In the former case, call that set $T_1$
and
repeat, yielding sets $T_0,T_1,\ldots,T_t$ where each $y \in T_s$ has at least
$\delta^2|T_{s-1}|/|X|$ out-degree into $T_{s-1}$.  Let $S$ be those vertices with
out-degree at least $\delta^3$ into $T_t$, so that $\delta|T_t| \leq |S| \leq
\delta^{-1}|T_t|$.
Let $\D' = \N^t_0$; by construction, each $y \in T_t$ has
\begin{align}
\mu_{M(\D')}(y, T_0) &\geq \prod_{s=0}^{t-1}\delta^2|T_s|/|X| \notag \\
 &\geq \delta^{2t-t(t-1)/2}|T_0|/|X| \notag \\
 &\geq \delta^{3+6k}\notag
\end{align}
Assuming case \ref{second} does not hold by setting $\D^* = \D'$, there is a set
$T\subseteq T_t$ of size at least
$(1-\delta^{5})|T_t|$
such that each $y \in T$ has $\mu_{M(\D')}(y, T_0 \cap \Ball[y;\ell]) \geq \delta^{3+6k}/2
=:\gamma'$.
It follows that at least $\delta^3|S| - \delta^{5}|T_t| \geq \delta^5|T_t|$ out-degree
from $S$ enters $T$.

Lemma \ref{nearby} implies each $y \in T$ is $((k-1)\sigma/4, \delta^{9}, \gamma'',
2\ell)$-covered by $\D''(\u) = \D', \D^k(\u)$ where $\gamma'' =
\gamma \gamma' \delta^{9}$ (we replace factors of $1/4$ with $\delta$).
Setting $\rho = 1-1/K$, lemma \ref{boost} implies each $y \in T$ is
$((k-3)\sigma/4, 1-2\delta^{9}, \gamma''', 2\ell)$-covered
by $\D'''(\u) = \D(\hat\u)$ for $\hat\u \sim_{1-1/K}\u$ where $\gamma''' =
\delta^{18K+1}\gamma''$.
Finally, since $\delta^5|T_t|/4|S| \geq \delta^7$, lemma \ref{extend} implies at least
$\delta^5|T_t|/2$ vertices in $S$ are
$((k+1)\sigma/4, \delta^7, \gamma''', \infty)$-covered by $\D^{k+1}(\u) = \u,
\D'''(\u)$, where
\[ \gamma''' = \delta^{18K+1+3+6k}\gamma/2 \geq 2\delta^{24K(k+1)} \]
Assuming case \ref{second} does not hold for $\D^* = \D^{k+1}$, at least $\delta^5|T_t|/4 \geq
\delta^{6(k+1)}|X|$ vertices in $S$ are $((k+1)\sigma/4, \delta^8, \delta^{24K(k+1)},
\ell)$-covered by $\D_{k+1}$.

Finally, note each $\D_k$ consists of a shuffling of $\N^k_{1-1/K}$ with $\N^{k'}_0$ where
$k' \leq 6k$ because the expanding case can occur at most $6k$ total times.
\end{proof}

To complete the proof of theorem \ref{ARVstrong}, recall $M$ is a $(\sigma,
\delta/4)$-uniform-matching-cover of $X$.  Let $1 \leq R \leq \log(n)/\log(1/\delta))$.
For $R <7$, lemma \ref{induct} implies a typical edge in $M$ has length
$\Omega(\sigma/\sqrt{\log(n/\delta)}) = \Omega(R\sigma/\sqrt{\log n})$
since $\log(n \geq \log(1/\delta))$ by assumption.  That is, setting $\D = \Sample(1,0)$
suffices.

For $R \geq 7$, set $K = \lfloor R/7\rfloor$ and $\ell = R\sigma/2^{10}\sqrt{\log(n)}$, so
that $\ell$ satisfies lemma \ref{induct}.  Lemma \ref{toolong} implies case \ref{first} of
lemma \ref{induct} can not hold for $K$, so case \ref{second} must hold.  That is, setting
$\D = \Sample(R, 1-1/K)$ suffices.
\ifFOCS
\else
\subsection{Using $\pm 1$ Coins}
One might be concerned with issues of precision required for sampling Gaussians.
Fortunately, it suffices to approximate them by sampling $\w \in \{\pm 1\}^k$
and returning $\frac{1}{\sqrt{k}}\sum_{i =1}^k \w_i$ for $k = O(\log n)$.
\begin{lemma}\label{binary} Suppose that instead of a random Gaussian $\u$, we sample a uniform random $\pm 1$
matrix $\U \in \R^{d \times k}$ and set $\u = \U\one$, where $\one \in \R^k$
has $\one_j = 1/\sqrt{k}$ for all $j \leq k$.
To sample a $\rho$-correlated copy $\hat\u$, we sample
$\hat\U \in \R^{d \times k}$ as a $\rho$-correlated copy of $\U$ (i.e., each
$\hat\U_{ij} = \U_{ij}$ with probability $\rho$ or a random $\pm 1$ with probability
$1-\rho$) and set $\hat\u = \hat\U\one$.  Then, for $k = O(R^2\log(1/\delta)) = O(\log n)$, theorem \ref{ARVstrong}
still holds.
\end{lemma}
The proof of lemma \ref{binary} is straightforward.  Lemmas \ref{goodu}, \ref{toolong}, \ref{nearby}
all still hold with similar constants even for $k=1$ (see e.g. \cite{database}), so the
only issue is lemma \ref{boost}.  For the latter, lemma \ref{iso} also
holds for $\rho$-correlated $\pm 1$ variables, so the only change needed is in bounding $(\v_x - \v_y)
\cdot (\u - \hat\u)$, which is easily done for $k = O(R^2\log(1/\delta)) = O(\log(n))$ using Bernstein's
inequality(see e.g. \cite{Boucheron}).
For completeness, we include the details in appendix \ref{plusminus}.
\fi
\section{\label{CUTMATCH}Lower-bound for the Cut-Matching Game}
Khandekar, Rao, and Vazirani proposed a primal-dual framework based on
the following two-player game game, which proceeds for $T$ rounds\cite{KRV}.
On each round, the cut player chooses a bisection $(S^t,\overline{S^t})$ of the vertices,
and the matching player responds with a perfect matching $M^t$ pairing each $x \in S^t$ with some
$y \in \overline{S^t}$  The payoff to the cut player is $h(H^T)$, where $H^t
= M^1 + \cdots + M^t$.  Thus on round $t$, the cut player aims to choose a cut so that
\emph{any} matching response $M^t$ will increase the expansion of $H^{t}$.

To see the connection to {\sc sparsest cut}, suppose the cut player has a strategy
that guarantees $h(H^T) \geq T/\kappa$, and consider a matching player that
plays as follows.
When given a bisection $(S,\overline{S})$,
the matching player connects a source $s$ to all $x\in S$ with edges of unit capacity and
a sink $t$ to all $y \in \overline{S}$.
A simple lemma similar to lemma \ref{cutval} implies that if the min-cut is at most $n/2$,
then it has expansion at most one.  Otherwise, the added edges are saturated, and assuming all
edges have integral capacities, the flow can be pseudo-decomposed into a matching; the
matching player responds with that matching.
Then, after $T$ rounds, we have either found a cut of expansion one or else routed $H^T$ in
$G$ with congestion $T$.  Assuming the cut-player forced $h(H^T) \geq
T/\kappa$, scaling down by $T$ yields a feasible flow routing a graph of expansion
$1/\kappa$, yielding a $\kappa$-approximation.

The following theorems appear in \cite{OSVV}.
\begin{theorem}[\cite{OSVV}] The cut player has an (efficient) strategy to ensure,
\begin{equation}
\exp\left(-\lambda_2(\L_{H^t})\right) \leq n\exp\left(\frac{-t}{O(\log n)}\right) \notag
\end{equation}
In particular, after $T = O(\log^2(n))$, the cut player can ensure $\lambda_2(\L_{H^T})
\geq \Omega(\log n)$, yielding an $O(\log n)$ factor approximation using $O(\log^2n)$
max-flows.
\end{theorem}
\begin{theorem}[\cite{OSVV}]\label{oldlower} The matching player can ensure
\begin{equation}
h(H^t) \leq O\left(\frac{1}{\sqrt{\log n}}\right) \cdot t \notag
\end{equation}
\end{theorem}
We prove the following.
\begin{theorem} \label{lower} The matching player can ensure
\begin{equation}
\lambda_2(\L_{H^t}) \leq O\left(\frac{\log\log n}{\log n}\right)\cdot t \notag
\end{equation}
\end{theorem}
Theorem \ref{lower} does not entirely eliminate the possibility of achieving a better approximation
in the cut-matching game, and indeed it is known among experts that there \emph{exists} an
(inefficient) strategy for the cut-player to ensure $\exp(-h(H^t)) \leq
n\exp(-t/O(\sqrt{\log(n)})$\cite{Lorenzo}.  However, theorem \ref{lower} says
that doing so will require certifying expansion via something stronger than
$\lambda_2(\L_{H^t})$.
For example, one could route another expander flow $H'$ in $H^T$ and certify
$h(H^T) \geq \lambda_2(\L_{H'})/2$.  Such an approach seems somewhat awkward though, as any such
flow might as well have been routed in $G$ directly.

In theorem \ref{oldlower}, the matching player arbitrarily identifies the vertices of $G$
with a hypercube, and tries to keep the dimension cuts sparse.  In particular, it is shown
that for any bisection $(S,\overline{S})$, there must always exist a matching that
raises the expansion of the average dimension cut by at most $O(1/\sqrt{d})$.

To prove theorem \ref{lower}, we identify the vertices of $G$ arbitrarily with a
dense set of points $\v_1,\ldots,\v_n\in \R^d$ on the sphere $S^{d-1}$, where $d = \Omega(\log(n)/\log\log(n))$.
Letting $\w_1,\ldots,\w_d \in \R^n$ be the column vectors of the $n \times d$ matrix with
row vectors $\{\v_x\}$, we show that for any bisection $(S,\overline{S})$ there must be a
matching that raises the average Rayleigh quotient $\frac{\w_i^T \L_{H} \w_i}{\w^T_i
\w_i}$ by at most $O(1/d)$.

The following lemma is an easy generalization of one in \cite{OSVV}.
\begin{lemma} \label{potential} Let $\v_1,\ldots, \v^n \in R^d$, and let $\w_1,\ldots,\w_d \in \R^n$ be
defined by $\w_{i,x} = \v_{x, i}$.  Define,
\begin{align}
\psi(t) &= \frac{1}{d}\sum_{i=1}^d \frac{\w_i^T \L_{H^t} \w_i}{\w^T_i \w_i} \notag
\end{align}
If all $\|\w_i\|^2 \geq L > 0$, then,
\begin{align}
\psi(t) - \psi(t-1) &\leq \frac{1}{dL} \sum_{xy \in M^t} \|\v_x - \v_y\|^2 \notag
\end{align}
\end{lemma}
\ifFOCS
\else
\begin{proof}
\begin{align}
\psi(t) - \psi(t-1) &=  \frac{1}{d}\sum_{i=1}^d \frac{\w_i^T\L_{M^{t}}\w_i}{\w_i^T\w_i}
\notag \\
 &\leq  \frac{1}{dL}\sum_{i=1}^d \sum_{xy \in M^t} (\w_{i,x} - \w_{i, y})^2 \notag \\
 &=  \frac{1}{dL}\sum_{xy \in M^t} \|\v_x - \v_y\|^2 \notag \qedhere
\end{align}
\end{proof}
\fi
If $\w_1,\ldots,\w_d$ are as in lemma \ref{potential} and all orthogonal to the all-1s
vector, then $\lambda_2(\L_{H^t}) \leq \psi(t)$; the orthogonality condition is equivalent to
$\sum_{x} \v_x = 0$.  Having fixed such an embedding, when presented with a bisection $(S,
\overline{S})$, the matching player aims to match points so as to minimize the average
distance between matched points.
The analysis of \cite{OSVV} shows that for the hypercube embedding $\{-1,1\}^d$, one can
obtain $\psi(t) - \psi(t-1) \leq O(1/\sqrt{d})$.  The analysis is not constructive; rather, they
use the vertex isoperimetry of the hypercube to establish an upper bound on the value
of the matching problem's dual LP, and then conclude a matching achieving that bound
exists by strong duality.  Their argument also depends on the
fact that for the hypercube embedding, the squared distances $\|\v_x-\v_y\|^2$ form a
metric.

In fact, the metric assumption is not needed, and there is also no need to apply LP duality.
We give a simple proof that large vertex isoperimetry of the embedding implies the simple greedy
strategy of iteratively matching closest points works.
\begin{lemma} \label{greed} Let $\v_1,\ldots,\v_n \in \R^d$ be a set of points such that, for any $S
\subseteq \{\v_1,\ldots,\v_n\}$ with $|S| \leq n/2$, $|\Ball[S;\sqrt{r}]| \geq (1+\Omega(1))|S|$,
Then, the greedy strategy produces $M$ with,
\begin{align}
\sum_{xy \in M} \|\v_x - \v_y\|^2 \leq O(nr) \notag
\end{align}
\end{lemma}
\ifFOCS
\IEEEproof
\else
\begin{proof}
\fi
Starting with $S,\overline{S}$, we pick $x \in S, y \in \overline{S}$ minimizing $\|\v_x -
\v_y|^2$, match them, and then remove them.  Repeated application of the
the isoperimetric condition implies that, for
all $S$ with $|S| \leq n/2$, $|\Ball[S;t\sqrt{r}]| \geq \min\{1 + n/2,
(1+\Omega(1))^t|S|\}$.  It follows that if two sets $A, B$ have size $s$, there must be
$x \in A$, $y \in B$ with $\|\v_x - \v_y\| \leq 2t\sqrt{r}$ for $t =
\lceil \log_{(1+\Omega(1))}(n/2s)\rceil + 1 \leq 2 + O(1)\log(n/2s)$.
Then, the total cost of the greedy solution is at most,
\begin{align}
\sum_{xy \in M}\| \v_x - \v_y\|^2 &\leq O\left( \sum_{s=1}^{n/2} \left(1 +
\log(n/2s)\right)^2 \cdot r\right) \notag \\
 &\leq O\left(n + \int_0^{n/2}\log^2(n/2s)\ \mathrm{d}s\right)r \notag \\
 &\leq O(nr) \notag \ifFOCS\else \qedhere \fi
\end{align}
\ifFOCS
\else
\end{proof}
For the case of theorem \ref{oldlower}, let $\v_x \in
\{-1/\sqrt{d},1/\sqrt{d}\}^d$ be the hypercube embedding and take
$L = n/d$ in lemma \ref{potential}.  The vertex isoperimetry of the hypercube implies $r =
O(1/\sqrt{d})$ in lemma \ref{greed}, yielding a strategy to ensure $\psi(t) \leq
O(nr/dL)\cdot t = O(1/\sqrt{d}) \cdot t$.
\fi

To prove theorem \ref{lower}, we choose $\v_x$ as per the following lemma, and take $L
= \Omega(n/d)$, $r = O(1/d)$, yielding a strategy to ensure $\psi(t) \leq O(nr/dL)
\cdot t=
O(1/d) \cdot t$.
\begin{lemma}\label{pointset} For every $d$, there exists a set of $n = O(\sqrt{d})^d$ points
$\v_1,\ldots,\v_n \in S^{d-1}$ such that $\sum_{i=1}^n \v_i = 0$, every $i \leq d$ has
$\sum_{x=1}^n \v_{x,i}^2 = \Omega(n/d)$, and for every $S \subseteq [n]$ with $|S| \leq n/2$,
$|\Ball[S;O(1/\sqrt{d})]| \geq (1+\Omega(1))|S|$.
\end{lemma}
The proof of lemma \ref{pointset} is a straightforward application of a construction of
Feige and Schechtman\cite{FS}, which we
\ifFOCS
\else include in appendix \ref{POINTSETPROOF} \fi
\section{\label{CONCLUSION}Final Remarks}
It will be interesting to see if efficient algorithms can be designed for the {\sc generalized sparsest cut}
problem\ifFOCS\else, where
we are given graphs $G$ and $H$ and aim to find a cut $(S,\overline{S})$ minimizing
$\frac{\sum_{x \in S, y \in \overline{S}} G_{xy}}{\sum_{x \in S, y \in \overline{S}}
H_{xy}}$ (when $H$ is the complete graph, the problem is essentially the regular {\sc
sparsest cut} problem, up to a factor of two).  \fi  The results of \cite{ALN} imply an
$O(\sqrt{\log n}\log\log n)$-approximation can be found by rounding a SDP similar to
\eqref{SDP}, but to the best of our knowledge no efficient algorithms have been designed
to approximately solve that SDP.

The boosting step in our proof of theorem \ref{ARVstrong} crucially depends on use of
the \emph{noise operator}.  Many hardness of approximation reductions for CSPs also
make use of that operator in their soundness analysis; what is the connection between
how it is used in each case?

Another question concerns the relation between the expander flow SDP and the original ``stronger''
SDP proposed by Goemans.  Constructing integrality
gaps for the latter is a notoriously hard problem.  Might it be any easier to construct
them for \eqref{SDP}?  If not, can one always ``round'' an embedding for the dual of
\eqref{SDP} to an embedding satisfying the triangle inequality constraints of
Goemans' program?
\section*{Acknowledgement}
We thank Umesh Vazirani and Satish Rao for helpful discussions, Ryan
O'Donnell for suggesting \cite{Borell, MORSS} to prove lemma \ref{iso}, and James Lee
for suggesting \cite{FS} to prove lemma \ref{pointset}.

\ifFOCS
\else
\bibliographystyle{plain}
\fi
\bibliography{flows}

\begin{thebibliography}{10}

\bibitem{database}
Dimitris Achlioptas.
\newblock Database-friendly random projections: Johnson-lindenstrauss with
  binary coins.
\newblock {\em J. Comput. Syst. Sci.}, 66(4):671--687, 2003.

\bibitem{Alon}
Noga Alon and V.~D. Milman.
\newblock $\lambda_1$, isoperimetric inequalities for graphs, and
  superconcentrators.
\newblock {\em J. Comb. Theory, Ser. B}, 38(1):73--88, 1985.

\bibitem{AHK}
Sanjeev Arora, Elad Hazan, and Satyen Kale.
\newblock $\,{O}(\sqrt {\log n})$ approximation to {\sc sparsest cut} in
  $\tilde{O}(n^2)$ time.
\newblock In {\em FOCS '04: Proceedings of the 45th Annual IEEE Symposium on
  Foundations of Computer Science}, pages 238--247, Washington, DC, USA, 2004.
  IEEE Computer Society.

\bibitem{MW}
Sanjeev Arora, Elad Hazan, and Satyen Kale.
\newblock The multiplicative weights update method: a meta algorithm and
  applications.
\newblock Technical report, Princeton University, 2005.

\bibitem{AroraKale}
Sanjeev Arora and Satyen Kale.
\newblock A combinatorial, primal-dual approach to semidefinite programs.
\newblock In {\em STOC '07: Proceedings of the thirty-ninth annual ACM
  symposium on Theory of computing}, pages 227--236, New York, NY, USA, 2007.
  ACM.

\bibitem{ALN}
Sanjeev Arora, James~R. Lee, and Assaf Naor.
\newblock Euclidean distortion and the sparsest cut.
\newblock In {\em STOC '05: Proceedings of the thirty-seventh annual ACM
  symposium on Theory of computing}, pages 553--562, New York, NY, USA, 2005.
  ACM.

\bibitem{ARV}
Sanjeev Arora, Satish Rao, and Umesh Vazirani.
\newblock Expander flows, geometric embeddings and graph partitioning.
\newblock In {\em STOC '04: Proceedings of the thirty-sixth annual ACM
  symposium on Theory of computing}, pages 222--231, New York, NY, USA, 2004.
  ACM.

\bibitem{Sparse}
Andr\'{a}s~A. Bencz\'{u}r and David~R. Karger.
\newblock Approximating s-t minimum cuts in $\,{O}(n^2)$ time.
\newblock In {\em STOC '96: Proceedings of the twenty-eighth annual ACM
  symposium on Theory of computing}, pages 47--55, New York, NY, USA, 1996.
  ACM.

\bibitem{Borell}
Christer Borell.
\newblock Positivity improving operators and hypercontractivity.
\newblock {\em Mathematische Zeitschrift}, 180:225--234, 1982.

\bibitem{Boucheron}
St{\'e}phane Boucheron, G{\'a}bor Lugosi, and Olivier Bousquet.
\newblock Concentration inequalities.
\newblock In {\em Advanced Lectures on Machine Learning}, pages 208--240.
  Springer, 2003.

\bibitem{FS}
Uriel Feige and Gideon Schechtman.
\newblock On the optimality of the random hyperplane rounding technique for max
  cut.
\newblock {\em Random Struct. Algorithms}, 20(3):403--440, 2002.

\bibitem{Fleischer}
Lisa~K. Fleischer.
\newblock Approximating fractional multicommodity flow independent of the
  number of commodities.
\newblock {\em SIAM Journal on Discrete Mathematics}, 13:505--520, 2000.

\bibitem{Freund}
Yoav Freund and Robert~E. Schapire.
\newblock Adaptive game playing using multiplicative weights.
\newblock {\em Games and Economic Behavior}, 29(1-2):79--103, October 1999.

\bibitem{GoldbergRao}
Andrew~V. Goldberg and Satish Rao.
\newblock Beyond the flow decomposition barrier.
\newblock {\em J. ACM}, 45(5):783--797, 1998.

\bibitem{KRV}
Rohit Khandekar, Satish Rao, and Umesh Vazirani.
\newblock Graph partitioning using single commodity flows.
\newblock In {\em STOC '06: Proceedings of the thirty-eighth annual ACM
  symposium on Theory of computing}, pages 385--390, New York, NY, USA, 2006.
  ACM.

\bibitem{JamesLee}
James~R. Lee.
\newblock On distance scales, embeddings, and efficient relaxations of the cut
  cone.
\newblock In {\em SODA '05: Proceedings of the sixteenth annual ACM-SIAM
  symposium on Discrete algorithms}, pages 92--101, Philadelphia, PA, USA,
  2005. Society for Industrial and Applied Mathematics.

\bibitem{lectures}
Jiri Matousek.
\newblock {\em Lectures on Discrete Geometry}.
\newblock Springer-Verlag New York, Inc., Secaucus, NJ, USA, 2002.

\bibitem{MORSS}
Elchanan Mossel, Oded Regev, Jeffrey~E. Steif, and Benny Sudakov.
\newblock Non-interactive correlation distillation, inhomogeneous markov
  chains, and the reverse bonami-beckner inequality.
\newblock {\em Israel Journal of Mathematics}, 154, 2006.

\bibitem{Lorenzo}
Lorenzo Orecchia.
\newblock personal communication, 2009.

\bibitem{OSVV}
Lorenzo Orecchia, Leonard~J. Schulman, Umesh~V. Vazirani, and Nisheeth~K.
  Vishnoi.
\newblock On partitioning graphs via single commodity flows.
\newblock In {\em STOC '08: Proceedings of the 40th annual ACM symposium on
  Theory of computing}, pages 461--470, New York, NY, USA, 2008. ACM.

\bibitem{Plotkin}
Serge~A. Plotkin, David~B. Shmoys, and Eva Tardos.
\newblock Fast approximation algorithms for fractional packing and covering
  problems.
\newblock {\em Mathematics of Operations Research}, 20:257--301, 1995.

\bibitem{SleatorTarjan}
Daniel~D. Sleator and Robert~Endre Tarjan.
\newblock A data structure for dynamic trees.
\newblock {\em J. Comput. Syst. Sci.}, 26(3):362--391, 1983.

\end{thebibliography}
\ifFOCS
\else
\appendix
\section{\label{ISOPROOF}Proof of Lemma \ref{iso}}
For $f : \R^d \to \R_{\geq 0}$, let $\|f\|_p = \E[f^p]^{1/p}$, where the expectation is
over the multivariate standard normal distribution.  For $x \in \R^d$, we write $y
\sim_\rho x$ for a $\rho$-correlated copy of $u$.  The \emph{Ornstein-Uhlenback operator}
is defined by,
\[ T_\rho f(x) = \E_{y \sim_\rho x}[f(y)] \]
\begin{theorem}[Borell\cite{Borell}]\label{hyper} Let $f : \R^d \to \R_{\geq 0}$
and $-\infty < q \leq p \leq
1$.  If $0 \leq \rho^2 \leq (1-p)/(1-q)$, then
\[ \|T_\rho f\|_q \geq \|f\|_p \qquad \textrm{for }0 \leq \rho^2 \leq (1-p)/(1-q)\]
\end{theorem}
By a change of variables, lemma \ref{iso} is equivalent to,
\begin{align}
\Pr_{\u}[\Pr_{\hat\u} [\hat\u \in A] < \tau] < \frac{\tau^{1-\rho}}{\delta} \notag
\end{align}
Let $f$
indicate $A$, and set $p = 1-\rho, q=1-1/\rho$.  Note $q < 0 < p \leq 1$ satisfy theorem
\ref{hyper}, so
\begin{equation}
\|T_\rho f\|_q \geq \|f\|_p = \delta^{1/p} \notag
\end{equation}
Then, $\Pr_{\hat\u \sim_\rho \u}[\hat\u \in A] = T_\rho f(\u)$, and we have,
\begin{align}
\Pr[T_\rho f < \tau] &= \Pr[(T_\rho f)^q > \tau^q] \notag\\
 &< \|(T_\rho f)\|_q^q\tau^{-q}\notag \\
 &\leq \delta^{q/p} \tau^{-q} \notag \\
 &= \left(\frac{\tau^{1-\rho}}{\delta}\right)^{1/\rho} \notag
\end{align}
For $\tau^{1-\rho}/\delta \leq 1$, raising the last line to $\rho$ can't decrease its value.
In the other case, the result is trivial.
\section{\label{plusminus} Proof of Lemma \ref{binary}}
Lemma \ref{goodu} only uses the fact that for a vector $\v$ and standard normal $\u$,
$(\u\cdot \v)^2 \geq \Omega(\|\v\|^2)$ with probability $\Omega(1)$.  That property
still holds.
\begin{lemma} Let $\U \in \R^{d \times k}$ be a uniform random $\pm 1$ matrix, and let
$\v \in \R^d$ be a vector.  Then,
\[ \Pr\left[(\v \cdot \U\one)^2 \geq \|\v\|^2/4\right] \geq 1/5 \]
\end{lemma}
\begin{proof}  It suffices to consider a unit vector $\v$.  Let $Z = \v \cdot
\U\one$.  Then,
\[ \E[Z^2] = \E\left[\left(\sum_{i\leq d, j\leq k} \v_i
\frac{\U_{ij}}{\sqrt{k}}\right)^2\right] = \sum_{i_1,i_2,j_1,j_2} \v_{i_1} \v_{i_2}
\frac{\E\left[\U_{i_1j_1}\U_{i_2j_2}\right]}{k} = \sum_{i,j} \v_i^2/k = \|\v\|^2 = 1
\]
\[ \E[Z^4] = 
 \sum_{i_1,\ldots,i_4,j_1,\ldots,j_4}
\v_{i_1}\cdots \v_{i_4} 
\frac{\E\left[\U_{i_1j_1}\cdots\U_{i_4j_4}\right]}{k^2} \leq 3\sum_{i_1,j_1,i_2,j_2}
\v_{i_1}^2 \v_{i_2}^2/k^2 = 3\|\v\|^4 = 3
\]

Then, for any $t \leq 1/\lambda$, we have,
\[ \Pr\left[Z^2 < \lambda\right] \leq \Pr\left[\left(1-tZ^2\right)^2 >
\left(1-t\lambda\right)^2\right] <
\frac{\E\left[\left(1-tZ^2\right)^2\right]}{\left(1-t\lambda\right)^2} =
\frac{1-2t\E[Z^2]+t^2\E[Z^4]}{\left(1-t\lambda\right)^2}
\]
Taking $t = \lambda = 1/4$ yields,
\[ \Pr\left[Z^2 < 1/4\right] < \frac{1-1/2+3/16}{(15/16)^2}  < 1/5 \]
\end{proof}

The remaining lemmas require a Gaussian-like bound on stretch; for that, we'll use the
following theorem.
\begin{theorem}[Bernstein's Inequality]\label{bernstein} Let $X_1,\ldots,X_n$ be independent random variables with
$\E[X_i] = 0$ and $X_i \leq 1$.  Let $\sigma^2 = \sum_{i=1}^n \E[X_i^2]$.  Then, for any
$t > 0$,
\[ \Pr\left[ \sum_{i=1}^n X_i > t\sigma\right] \leq
\exp\left(\frac{-t^2}{2+t/3\sigma}\right) \]
\end{theorem}

The next lemma says that if $k$ is large enough, we can obtain Gaussian-like bounds on
stretch.  Note that when $\rho=0$ much better bounds are possible, in that even $k=1$
works (see \cite{database}).
\begin{lemma} \label{stretch}Let $\U \in \R^{d \times k}$ be an arbitrary $\pm 1$ matrix, and let
$\hat\U \sim_\rho \U$ be a $\rho$-correlated copy of $\U$.  Then, for any vector $\v$, and
any $0 < t \leq \sqrt{k(1-\rho^2)}$,
\[ \Pr\left[\v\cdot \hat\U\one > \rho(\v\cdot \U\one) + t\sqrt{1-\rho^2}\|\v\| \right] \leq e^{-t^2/3} \]
\[ \Pr\left[\v\cdot \hat\U\one < \rho(\v\cdot \U\one) - t\sqrt{1-\rho^2}\|\v\| \right] \leq e^{-t^2/3} \]
\end{lemma}
\begin{proof}
It suffices to consider a unit vector $\v$.
For each $i \leq d$, $j \leq k$, let $Z_{ij} = \v_i(\hat\u_{ij} - \rho \U_{ij})/2$, so that
we have $\E[Z_{ij}] = 0$, $|Z_{ij}| \leq 1$, and $\E[Z_{ij}^2] = (1-\rho^2)\v_i^2/4$.
Note that,
\[\v\cdot \hat\U\one = \rho(\v\cdot \hat\U\one) + \frac{2}{\sqrt{k}}
\sum_{i\leq d, j\leq k} Z_{ij} \]

Applying theorem \ref{bernstein} with $\sigma^2 = k(1-\rho^2)\|\v\|^2/4 = k(1-\rho^2)/4$, we
have
\[ \Pr\left[ \sum_{i,j} Z_{ij} > t\sigma \right] \leq
\exp\left(\frac{-t^2}{2+t/3\sigma}\right) \leq e^{-t^2/3} \]
proving the first part.  The second part follows by applying the same argument to
$-Z_{ij}$.
\end{proof}

For lemmas \ref{toolong} and \ref{nearby}, we use $\rho = 0$ and $t =
O(\sqrt{\log(1/\delta)})$, so $k = O(\log(1/\delta))$ suffices.
For lemma \ref{boost}, we use $\rho = 1-1/K$ and $t = O(\sqrt{K\log(1/\delta)})$, so $k =
O(K^2\log(1/\delta))$ suffices.  Also, lemma \ref{iso} holds for the uniform measure on
the hypercube, as Borell's theorem also holds for $f : \{-1,+1\}^n \to \R_{\geq 0}$.

\section{\label{POINTSETPROOF}Proof of Lemma \ref{pointset}}
\begin{lemma}[Feige, Schechtman \cite{FS}] \label{cells} For each $0 < \gamma < \pi/2$, the sphere $S^{d-1}$ can be partitioned into
$n = (O(1)/\gamma)^d$ equal volume cells, each of diameter at most $\gamma$.
\end{lemma}
Apply lemma \ref{cells} with $\gamma = 1/\sqrt{d}$, yielding cells $C_1,\ldots,C_n$ with
$n = \exp(O(d\log d))$.  Let $V$ be a set of $n$ arbitrary points, each from a distinct
cell; for convenience, let us choose $V$ so that $V \cap (-V) = \emptyset$.
\begin{claim} \label{meas} If $A \subseteq S^{d-1}$ has $\mu(A) \geq \alpha$, then $|\Ball(A;\gamma) \cap
V| \geq \alpha n$; if $A \subseteq V$ has $|A| \geq \alpha n$, then $\mu(\Ball(A;\gamma))
\geq \alpha$.
\end{claim} \begin{proof}  For the first direction, if $\mu(A) \geq \alpha$, $A$
intersects at least $\alpha n$ of the cells, so $\Ball[A;\gamma]$ contains at least $\alpha n$
cells, and hence at least $\alpha n$ elements of $V$.  For the second, if $A \subseteq V$
has size at least $\alpha n$, then $\Ball[A ; \gamma]$ contains at least $\alpha n$
cells, so $\mu(\Ball[A;\gamma]) \geq \alpha$. \end{proof}

\begin{claim}\label{longvar} For every $i \leq d$, $\sum_{\v \in V} \v_i^2 \geq \Omega(n/d)$.
\end{claim}
\begin{proof} Let $A = \{x \in S^{d-1} : x_i \geq \sqrt{2/d}\}$.  By bounds on the measure
of spherical caps (see e.g. \cite{lectures}), $\mu(A) \geq 1/12$.  Using claim \ref{meas},
$|\Ball(A;\gamma) \cap V| \geq (1/12)n$.  Then, since $x_i \geq \sqrt{2/d} - \gamma \geq
\sqrt{1/8d}$ for all $x \in \Ball(A;\gamma)$, we have  $\sum_{\v \in V} \v_i^2 \geq
(1/12)n(1/8d) \geq \Omega(n/d)$.\end{proof}

\begin{claim} \label{expand} For every $A \subseteq V$ with $|A| \leq n/2$, $|\Ball[A; O(1/\sqrt{d})] \cap X|
\geq (1+1/12)|A|$.
\end{claim}\begin{proof} Let $A \subseteq V$ have $|A| \leq n/2$, and set $A_1 = \Ball[A;
\gamma], A_2 = \Ball[A_1; 4/\sqrt{d}], A_3 =
\Ball[A_2;\gamma]$; the goal is to show $|A_3 \cap V| \geq (1+1/12)|A|$.  Note claim
\ref{meas} ensures $\mu(A_1) \geq |A|/n$.  If $\mu(A_1) \geq (1+1/12)/2$, then $\mu(A_2)
\geq (1+1/12)|A|/n$.  Otherwise, by the
isoperimetric inequality on the sphere (see e.g. \cite{lectures}), $\mu(A_2) \geq
(1+1/12)\mu(A_1) \geq (1+1/12)|A|/n$.  By by claim \ref{meas}, $|A_3 \cap V| \geq
(1+1/12)|A|$.
\end{proof}

Now to prove lemma \ref{pointset}, we let $V' = V \cup -V$.  Clearly $\sum_{\v \in V'}
\v = 0$, and claim \ref{longvar} still applies to $V'$, so it remains only to argue claim
\ref{expand} still holds for $V'$.
Let $A \subseteq V'$ have $|A| \leq
|V'|/2 = n$.  Let $A = A_+ \cup A_-$ where $A_+ \subseteq V$ and $A_- \subseteq -V$, and
suppose $|A_+| \leq |A_-|$ (in the other case, an analogous argument applies).  We
consider two cases.  First, if $|A_+| \leq |A_-|/2$, then $\mu(\Ball[A_-;\gamma]) \geq
|A_-|/n$, implying $|\Ball[A_-;2\gamma] \cap V| \geq |A_-|/n$.  Therefore, $|\Ball[A;2\gamma] \cap
V'| \geq 2|A_-| \geq (4/3)|A|$.  Otherwise, $|A_+| \leq n/2$ and $|A_+| \geq |A|/3$, so
claim \ref{expand} implies $|\Ball[A_+; O(1/\sqrt{d})] \cap V| \geq (1+1/12)|A_+|$.  Therefore,
$|\Ball[A; O(1/\sqrt{d})] \cap V'| \geq (1+1/12)|A_+| + |A_-| \geq (1+1/36)|A|$.
\fi
\end{document}